\documentclass[conference]{IEEEtran}

\usepackage{amssymb}
\usepackage{bbm}
\usepackage{amsmath}
\usepackage{amsthm}
\usepackage{MnSymbol}
\usepackage{mathrsfs}
\usepackage{mathtools}
\usepackage{algorithm}
\usepackage{algpseudocode}
\usepackage{subfig}
\usepackage{hyperref}
\usepackage{array,booktabs}
\usepackage{multirow}
\usepackage{afterpage}
\usepackage{xcolor}
\usepackage{array}

\newtheorem{lemma}{Lemma}

\newcommand{\pr}
{
\text{pr}
}

\begin{document}
\title{A Local Updating Algorithm for Personalized PageRank via Chebyshev Polynomials}

\author{\IEEEauthorblockN{E. Bautista$^{1}$, M. Latapy$^{1}$,}
\IEEEauthorblockA{
$^{(1)}$ Sorbonne Université, CNRS, LIP6, F-75005 Paris, France \\
\vspace*{-0.8cm}}
}

\maketitle

\begin{abstract}
The personalized PageRank algorithm is one of the most versatile tools for the analysis of networks. In spite of its ubiquity, maintaining personalized PageRank vectors when the underlying network constantly evolves is still a challenging task. To address this limitation, this work proposes a novel distributed algorithm to locally update personalized PageRank vectors when the graph topology changes. The proposed algorithm is based on the use of Chebyshev polynomials and a novel update equation that encompasses a large family of PageRank-based methods. In particular, the algorithm has the following advantages: (i) it has faster convergence speed than state-of-the-art alternatives for local PageRank updating; and (ii) it can update the solution of recent generalizations of PageRank for which no updating algorithms have been developed. Experiments in a real-world temporal network of an autonomous system validate the effectiveness of the proposed algorithm.

\end{abstract}
\IEEEpeerreviewmaketitle

\section{Introduction}

\subsection{Context}
The personalized PageRank algorithm is one of the great successes in network science. It has been used in a wide amount of applications that include ranking of websites \cite{Page_1999_Pagerank, Ding_2003_PageRank, Haveliwala_2003_topic}, clustering data \cite{Chung_2009_Local, Andersen_2006_Local, Tabrizi_2013_Personalized}, classifying objects \cite{Avrachenkov_2012_Classification, Merkurjev_2018_Semi, Dostal_2014_Exploration, Avrachenkov_2012_Generalized}, detection of anomalous events \cite{Fontugne_2019_BGP, Yoon_2019_Fast, Yao_2012_Anomaly}, or recommender systems \cite{AlJanabi_2019_Recommendation, Zhang_2009_Collaborative, Nguyen_2015_Evaluation}, to name a few. Additionally, the numerous theoretical properties of personalized PageRank \cite{Langville_2004_Deeper, Ipsen_2006_Mathematical, Brezinski_2006_Pagerank, Pretto_2002_Theoretical} have recently been leveraged to propose a large family of PageRank extensions used in the semi-supervised learning setting, allowing to tackle challenging settings such as signed graphs \cite{Bautista_2019_Gamma}, anomalous diffusion processes \cite{DeNigris_2017_Fractional}, infinite dimensional spaces \cite{Mai_2017_Counterintuitive}, Sobolev spaces \cite{Zhou_2011_Semi}, or time-graph dual spaces \cite{Girault_2014_semi}. However, despite the success of all these PageRank-based algorithms, they still pose problems from the algorithmic point of view. In particular, they do not adapt well to networks that evolve over time, a situation that we aim to address in this work.

\subsection{Related Works}
In its simplest form, personalized PageRank can be interpreted as the stationary distribution of a random walk \cite{Page_1999_Pagerank}. This interpretation has been exploited by numerous works to develop efficient algorithms for computing PageRank \cite{Haveliwala_1999_Efficient, Kamvar_2004_Adaptive, Fujiwara_2012_Efficient, Bahmani_2011_Fast, Maehara_2014_Computing, Avrachenkov_2007_Monte, Andersen_2006_Local, Berkhin_2006_Bookmark}. However, in cases where the graph evolves over time, such stationary distribution drifts and needs to be updated. While the algorithms above allow to efficiently recompute a personalized PageRank vector, proceeding in this way becomes impractical for very large networks that constantly evolve, such as the web graph that has been reported to have $60\times10^{12}$ nodes and evolve at a rate of $600\times 10^3$ new pages created every second \cite{Ohsaka_2015_Efficient}. Additionally, several works \cite{Ipsen_2006_Mathematical, Brezinski_2006_Pagerank, Pretto_2002_Theoretical} have studied the impact of graph perturbations to PageRank vectors, showing that (i) the magnitude of the change to a PageRank vector is upper bounded by the size of the perturbation (magnitude of the change to the graph transition matrix); and (ii) the PageRank vector mostly changes in the entries associated to nodes close to the perturbed area. This implies that if a perturbation is small, then it is wasteful to recompute the PageRank vector from scratch. A better alternative consists in only updating the scores of nodes close to the perturbation. In the literature, this challenge is commonly referred to as local PageRank updating. There are three reference methods. Firstly, \cite{Bahmani_2010_Fast} proposes a Monte-Carlo approach where multiple random walkers are run to estimate the entries of the PageRank vector based on how frequently walkers visit nodes. Secondly, the work of \cite{Yoon_2018_Fast} exploits the random walk with restart interpretation of personalized PageRank to show that the stationary distribution of this walk can be updated by running a local diffusion process in the affected area. Then, \cite{Yoon_2018_Fast} runs this local process in a distributed fashion via the power method. Thirdly, the work of \cite{Ohsaka_2015_Efficient} proposes a centralized updating algorithm based on the use of a residual and an approximation vector, sequentially pushing mass from the residual vector into the approximation vector using Gauss-Southwell update rules. While these works have been subject of deep theoretical studies \cite{Zhang_2016_Approximate} and successfully applied in practice \cite{yoon2020autonomous}, they still suffer from two main limitations: (i) they have slow convergence rates, particularly when trying to attain very small approximation errors; and (ii) they fully rely on the random walk interpretation of PageRank and thus cannot be used to update the novel generalizations of PageRank used in semi-supervised learning \cite{Bautista_2019_Gamma, DeNigris_2017_Fractional, Mai_2017_Counterintuitive, Zhou_2011_Semi, Girault_2014_semi} which rely on more complex dynamical processes. 

\subsection{Goals, Contributions and Outline}
In this work, we address the two limitations listed above. We propose a novel local updating algorithm based on Chebyshev polynomials. These polynomials have already been used to efficiently approximate personalized PageRank vectors from scratch \cite{Bautista_2019_Laplacian}, showing faster convergence than Gauss-Southwell and power iteration based methods (the building blocks of \cite{Ohsaka_2015_Efficient} and \cite{Yoon_2018_Fast}, respectively). However, Chebyshev polynomials have not yet been considered in the updating setting because they lack the ability to set an initial guess, which all the updating algorithms described above rely upon. Thus, the core of our proposal relies on a novel updating equation that is tailored for the large family of PageRank methods mentioned above. Concretely, it allows us to (i) use a previous personalized PageRank vector as an initial guess in the Chebyshev context; (ii) cast the updating challenge as the task of running a local diffusion process that we efficiently compute via the Chebyshev polynomials; (iii) update any of the recent generalized formulations of PageRank \cite{Bautista_2019_Gamma, DeNigris_2017_Fractional, Mai_2017_Counterintuitive, Zhou_2011_Semi, Girault_2014_semi} among which classical personalized PageRank arises as a special case.

The paper is organized as follows: Section \ref{sec2} sets definitions and reviews local updating methods. Section \ref{sec3} presents Chebyshev polynomials and introduces the proposed algorithm. Section \ref{sec4} numerically evaluates the algorithm. Section \ref{sec5} concludes the work.

\section{Definitions and State-of-the-art}\label{sec2}
\subsection{Definitions}
Let $\mathcal{G} (\mathcal{V}, \mathcal{E}, W)$ be an undirected weighted graph where $W$ refers to the adjacency matrix, $\mathcal{E}$ is the set of edges and $\mathcal{V}$ is the set of vertices. By $D = diag(d_1, \cdots, d_N)$ we denote the diagonal matrix of degrees, where $d_i = \sum_{j}W_{ij}$. The combinatorial Laplacian and the random walk transition matrices are denoted by $L = D - W$ and $P = D^{-1}W$, respectively. We refer by $u \sim v$ that $u$ is adjacent to $v$. Additionally, let $\pr_{\alpha} (y)$ denote the personalized PageRank vector of $\mathcal{G}$ with restarting probability $\alpha$ and initial condition $y$. Personalized PageRank is defined as the solution to the fixed point equation
\begin{equation}\label{eq.pagerank.fixedpoint}
\pr_{\alpha} (y) = (1-\alpha) y + \alpha P^T \pr_{\alpha}(y).
\end{equation}
Let $\widetilde{\mathcal{G}}(\widetilde{\mathcal{V}},  \widetilde{\mathcal{E}}, \widetilde{W} )$ denote an undirected weighted graph which is an evolved or perturbed version of $\mathcal{G}$. In analogy to $\mathcal{G}$, the matrices $\widetilde{D}$, $\widetilde{L}$ and $\widetilde{P}$ refer to the degree, Laplacian and random walk transition matrices of $\widetilde{\mathcal{G}}$, respectively. Additionally, we let $\widetilde{\pr}_{\alpha}(y)$ be the personalized PageRank vector of $\widetilde{\mathcal{G}}$ with restarting probability $\alpha$ and initial condition $y$. The updating challenge refers to the problem of estimating $\widetilde{\pr}_{\alpha}(y)$ from $\pr_{\alpha}(y)$. To have consistent sized matrices, we model nodes joining/leaving the network as isolated nodes that get connected/disconnected. By definition, isolated nodes correspond to zero rows and columns in the graph matrices and their degree and inverse degree are zero. 

\subsection{State-of-the-art Approaches for Local PageRank Updating}
In this section, we briefly review the state-of-the-art methods of \cite{Yoon_2018_Fast} and \cite{Ohsaka_2015_Efficient} for local PageRank updating. We focus on these two approaches because they are deterministic methods that can update a PageRank vector up to any desired accuracy, hence they directly compare to our proposed algorithm that we describe in Section \ref{sec3}.

{\bf Random Walk with Restart (RWR) \cite{Yoon_2018_Fast}}. This work is rooted in the power iteration method, which guarantees that the recursive formula
\begin{equation}\label{pr.powerit.recursion}
\widetilde{p}^{(t)} = (1-\alpha) y + \alpha \widetilde{P}^T \widetilde{p}^{(t-1)}
\end{equation} converges to $\widetilde{\pr}_\alpha(y)$ when $t \to \infty$. The authors take benefit of the possibility to set an initial guess $\widetilde{p}^{(0)}$, also known as warm restart. They set $\widetilde{p}^{(0)} = \pr_\alpha(y)$. Since it is guaranteed that a PageRank vector does not substantially change under small perturbations, this warm restart sets the trajectory of the recursive equation (\ref{pr.powerit.recursion}) very close to the fixed point, which helps to drastically reduce the number of iterations needed to obtain a good approximation. Notice however that, even though this reduces the number of iterations, the initial guess $\widetilde{p}^{(0)}$ is a dense vector (most entries are non-zero). Thus, computations are needed for all vertices in the graph albeit the update only occurs locally. The authors of \cite{Yoon_2018_Fast} address this issue by showing that (\ref{pr.powerit.recursion}), under a warm restart, can be rewritten as:
\begin{equation}\label{pr.powerit.update}
\widetilde{\pr}_\alpha(y) = \pr_\alpha(y) + \frac{1}{(1-\alpha)} \widetilde{\pr}_{\alpha}(r),
\end{equation}
where $r =  \alpha [ \widetilde{P}^T - P^T ] \pr_{\alpha}(y)$. Eq. (\ref{pr.powerit.update}) shows that updating an existing PageRank vector amounts to computing another PageRank vector with an initial seed $r$ that is completely localized (only non-zero) in the 1-hop vicinity of the perturbed nodes. Clearly, $\widetilde{\pr}_{\alpha}(r)$ can be computed by setting $y = r$ and $\widetilde{p}^{(0)} = r$ in (\ref{pr.powerit.recursion}), showing that running the recursion for a few iterations is essentially equivalent to locally diffusing $r$ to the affected nodes and updating their values.

{\bf Push method \cite{Ohsaka_2015_Efficient}}. This work is rooted in the Gauss-Southwell recursive formula which uses two vectors to approximate $\widetilde{\pr}_\alpha(y)$: an approximation vector $\widetilde{p}$ and a residual vector $\widetilde{r}$ coding the difference between $\widetilde{\pr}_{\alpha}(y)$ and $\widetilde{p}$. The method starts with an initial guess $\widetilde{p}^{(0)}$ and then, at iteration $t$, transfers mass from $\widetilde{r}^{(t -1)}$ into $\widetilde{p}^{(t)}$ such that the following invariant is preserved 
\begin{equation}\label{push.invariant.eq}
\widetilde{\pr}_\alpha(y) = \widetilde{p}^{(t)} + \frac{1}{1- \alpha} \widetilde{\pr}_\alpha(\widetilde{r}^{(t)}). 
\end{equation}
Clearly, minimizing the entries of $\widetilde{r}$ imply that $\widetilde{p}$ converges to $\widetilde{\pr}_\alpha(y)$. To attain this, if at iteration $t$ the largest entry of $\widetilde{r}^{(t)}$ corresponds to vertex $u$, then the state of the algorithm at iteration $t+1$ is determined by the following set of update equations:
\begin{align}
\widetilde{p}^{(t+1)} &= \widetilde{p}^{(t)} + \widetilde{r}^{(t)}_u \delta_u. \\
\widetilde{r}^{(t+1)} &= \widetilde{r}^{(t)} - \widetilde{r}^{(t)}_u \delta_u + \alpha \widetilde{P}^T \widetilde{r}^{(t)}_u \delta_u.
\end{align}
This procedure is repeated until all entries from the residual diminish below some user-defined threshold. The authors of \cite{Ohsaka_2015_Efficient} show that setting $\widetilde{p}^{(0)} = \pr_{\alpha}(y)$ and $\widetilde{r}^{(0)} = \alpha [  \widetilde{P}^T - P^T ]\pr_{\alpha}(y)$ preserve the invariant (\ref{push.invariant.eq}). While $\widetilde{r}^{(0)}$ is similar to the residual of the RWR method, in the push method it is used differently: the locality of $\widetilde{r}^{(0)}$ implies that only a few entries can surpass the tolerance threshold, meaning that only a few push operations are needed to drive them below the threshold again and obtain a good approximation of the evolved PageRank vector.

\section{Proposed Method}\label{sec3}
\subsection{PageRank Computation via Chebyshev Polynomials}\label{sec3.1}
\begin{algorithm}[t]
\caption{\cite{Shuman_2011_Chebyshev} Distributed application of a matrix function to a vector via Chebyshev polynomials}
\label{ChebyshevPoly.alg}
\begin{algorithmic}
        \State \textbf{Input at node u}: $\lambda_{max}$, $y_u$, $\mu$, $\mathcal{R}_{uv} ~ \forall ~ v \sim u$, $K$, and $\{c_t: t = 0, \dots, K\}$. 
        \State \textbf{Output at node u}: $f^{(K)}_u$
        \State $(\bar{T}_0(\mathcal{R}) y)_u = y_u$
        \State Transmit $y_u$ to all neighbours $v \sim u$
        \State Receive $y_v$ from all neighbours $v \sim u$
        \State $(\bar{T}_1(\mathcal{R}) y)_u = \sum\limits_{v = \{v \sim u\} \cup u} \frac{1}{\phi} \mathcal{R}_{uv} y_v - y_u $
        
        \For{$t = 2 : K$} \vspace{5pt}
        \State Transmit $(\bar{T}_{t-1}(\mathcal{R}) y)_u$ to all neighbours $v \sim u$ \vspace{5pt}
        \State Receive $(\bar{T}_{t-1}(\mathcal{R}) y)_v$ from all neighbors $v \sim u$ \vspace{5pt}
        \State $(\bar{T}_t(\mathcal{R}) y)_u = \sum\limits_{v = \{v \sim u\} \cup u} \frac{2}{\phi} (\bar{T}_{t-1}(\mathcal{R}) y)_v - 2 (\bar{T}_{t-1}(\mathcal{R}) y)_u  - (\bar{T}_{t-2} (\mathcal{R})y)_v  $
        \EndFor \vspace{5pt}
        \State Return $f^{(K)}_u = \frac{1}{2}c_0 y_u + \sum_{t = 1}^K c_t (\bar{T}_t(\mathcal{R}) y)_u  $
\end{algorithmic}
\end{algorithm}
In this subsection, we present the Chebyshev polynomials, which are a general technique to approximate matrix functions and form the basis of our algorithm derived in Section \ref{sec3.2}. As discussed in the introduction, the two main drawbacks of current local PageRank updating algorithms are that (i) they are slow to converge; and (ii) they cannot address the novel generalizations of personalized PageRank that no longer rely on random walk processes. Notably, the Chebyshev polynomials carry the potential to address these limitations because (i) they can account for general operators, thus covering the recent generalization of PageRank in the literature; and (ii) they converge faster than methods based on power-iteration and Gauss-Southwell rules when computing PageRank vectors from scratch \cite{Bautista_2019_Laplacian}. 

In the context of signal processing on graphs, the Chebyshev polynomials were introduced in \cite{Shuman_2011_Chebyshev} as a mean to approximate functions of a graph matrix, achieving considerable success in the contexts of graph signal filtering \cite{cheng_2019_iterative, tseng_2021_minimax, tian_2014_chebyshev} and graph neural networks \cite{defferrard_2016_convolutional, yan_2021_spatial}. They operate as follows: let $\mathcal{R}$ denote a graph matrix and $h(\mathcal{R})$ be a function of it. Then, \cite{Shuman_2011_Chebyshev} shows that $h(\mathcal{R})$ can be approximated by means of the truncated series:
\begin{equation}\label{}
h(\mathcal{R}) \approx \frac{1}{2}c_0 + \sum_{t = 1}^K  c_t \bar{T}_t(\mathcal{R}),
\end{equation}
where
\begin{equation}\label{Chap4_ChebyPol.eq}
\bar{T}_t(\mathcal{R}) = \begin{cases}
1, & t = 0 \\
\frac{\mathcal{R} - \phi}{\phi}, & t = 1 \\
2\left(\frac{\mathcal{R} - \phi}{\phi}\right) \bar{T}_{t-1} - \bar{T}_{t-2}, & t \geq 2
\end{cases}
\end{equation}
$\phi = \lambda_{max}/2$, $c_t = \frac{2}{\pi} \int_0^{\pi} cos(t \theta) h(\phi (cos(\theta)+1) ) d\theta$, and $\lambda_{max}$ is the spectral radius bound of $\mathcal{R}$. Eq. (\ref{Chap4_ChebyPol.eq}) is known as the Chebyshev polynomial approximation of $h(\mathcal{R})$. One of its assets is that it allows to approximate the result of multiplying $h(\mathcal{R})$ with a vector $y$ in a distributed fashion. To see this, let us recall that $y$ can be interpreted as a signal that lives on the vertices of the graph encoded by $\mathcal{R}$. Thus, if nodes are given communication and computation capabilities, each node can compute its own value of $h(\mathcal{R}) y$ by transmitting and receiving messages to and from their neighbors. This distributed algorithm is detailed in Algorithm \ref{ChebyshevPoly.alg}. While \cite{Bautista_2019_Laplacian} shows that Chebyshev polynomials can converge to personalized PageRank vectors significantly faster than power iteration and Gauss-Southwell methods from scratch, they have not been considered in the updating setting because they do not offer the possibility to set an initial guess.

\subsection{Local PageRank updating via Chebyshev polynomials}\label{sec3.2}
\begin{table*}[t!]
\footnotesize
  \centering
    \begin{tabular}{|c|c|c|c|c|c|c|c|}\toprule 
{\bf Method} & Std. PageRank \cite{Page_1999_Pagerank} & $L^\gamma$-PageRank \cite{Bautista_2019_Gamma} & Iter. PageRank \cite{Zhou_2011_Semi} & Recentered kernel \cite{Mai_2017_Counterintuitive} & Time-graph dual \cite{Girault_2014_semi} & Anom. Diffusion \cite{DeNigris_2017_Fractional} \\ \midrule
{\bf{$\mathcal{R} $}} & $LD^{-1}$ & $ L^{\gamma}D_{\gamma}^{-1}$  & $(L D^{-1})^m$ & $-PWP$, $P = \mathbb{I} - \frac{1}{N} \mathbbm{1}\mathbbm{1}^T$ & $D^{-\sigma} L D^{\sigma - 1}$ & $D_{\gamma}^{-\sigma} L^{\gamma} D_{\gamma}^{\sigma - 1}$ \\ \bottomrule
 \end{tabular}
\caption{Possible choices of the reference operator $\mathcal{R}$ reported in the literature.}
\label{table.operators}
\end{table*}
In this subsection, we detail our main contribution: a Chebyshev polynomial-based local updating algorithm tailored for a large family of PageRank methods. We start by noticing that standard personalized PageRank and the novel generalizations used in semi-supervised learning \cite{Bautista_2019_Gamma, DeNigris_2017_Fractional, Mai_2017_Counterintuitive, Zhou_2011_Semi, Girault_2014_semi} can all be framed under one same formalism in terms of matrix functions. To show this, let us introduce the change of variable $\alpha = \frac{1}{\mu + 1}$. It is easy to see that the PageRank fixed point equation (\ref{eq.pagerank.fixedpoint}) can be rewritten as the partial differential equation: $LD^{-1}\pr_{\mu}(y) + \mu \pr_{\mu}(y) = \mu y$, where $LD^{-1} = (\mathbb{I} - P^T)$ is the so-called random walk Laplacian. This expression implies that PageRank can alternatively be interpreted as the equilibrium state of a dynamical process driven by a discrete Helmholtz equation in which the dynamics are ruled by the operator $LD^{-1}$. Notably, several of the novel generalizations of personalized PageRank admit the same interpretation, allowing their solutions to be expressed in the following general form:
\begin{equation}\label{eq.helmholts}
    \mathcal{R}\pr_{\mu}(y) + \mu \pr_{\mu}(y) = \mu y,
\end{equation}
where $\mathcal{R}$ denotes a generalized reference operator associated to graph $\mathcal{G}$. Thus, the only difference among several personalized PageRank generalizations is the choice of the operator $\mathcal{R}$. In Table \ref{table.operators}, we list some of the possible choices of $\mathcal{R}$ and the methods associated to them.

Clearly, the advantage of Eq. (\ref{eq.helmholts}) is that any algorithm that we derive based on it automatically covers a large family of PageRank methods. Therefore, the updating algorithm we propose in this work is an algorithm to update the solution of Eq. (\ref{eq.helmholts}). To derive our updating algorithm, we start by noticing that the solution of Eq. (\ref{eq.helmholts}), for an evolved graph $\widetilde{\mathcal{G}}$, can be expressed as a matrix function of $\widetilde{\mathcal{R}}$ in the following way
\begin{equation}\label{sol.proposed.scratch}
\widetilde{pr}_{\mu}(y) = \mu(\widetilde{\mathcal{R}} + \mu \mathbb{I})^{-1}y.
\end{equation}
While Eq. (\ref{sol.proposed.scratch}) can be leveraged to efficiently compute $\widetilde{pr}_{\mu}(y)$ from scratch (for instance via the Chebyshev polynomials), it is not useful in an updating scenario because it does not allow to set an initial guess. Therefore, our first goal is to derive a recursive equation that converges to (\ref{eq.helmholts}), which we can then use to set $\pr_{\mu}(y)$ as an initial guess. A natural way to proceed is by developing Eq. (\ref{sol.proposed.scratch}) in its geometric series. However, we stress that this approach results in a recursive expression $\widetilde{p}^{(t)} = y + \frac{1}{\mu} \widetilde{\mathcal{R}} \widetilde{p}^{(t-1)}$ that only converges to $\widetilde{\pr}_{\mu}(y)$ when $\mu > \lambda_{max}$. To obtain a recursive equation that converges for all $\mu > 0$ (hence for all $\alpha \in (0, 1]$), we perform the following transformations: (i) we normalize the spectrum (eigenvalues) of $\widetilde{\mathcal{R}}$ to the range $[-1, 1]$; and (ii) we map the spectral domain of the matrix function $h(\widetilde{\mathcal{R}}) = \mu(\widetilde{\mathcal{R}}+ \mu \mathbb{I})^{-1} $ to the range $[-1, 1]$. By applying transformation (i), we obtain the new operator:
\begin{equation}\label{eq.renormalization}
\widetilde{\mathcal{S}} = (2/\lambda_{max}) \widetilde{\mathcal{R}} - \mathbb{I}.
\end{equation}
We refer to the diagonal matrices of eigenvalues of $\widetilde{\mathcal{R}}$ and $\widetilde{\mathcal{S}}$ by $\Lambda$ and $s$, respectively. Then, by applying transformation (ii) we obtain a new matrix function now depending on $s$:
\begin{align}\label{proposed.transf.ii}
h(\Lambda) &= \frac{\mu}{\Lambda + \mu} \\
            &= \frac{\mu}{ (\lambda_{max}/2)(s + 1) + \mu} \\ 
            &=  \frac{2 \mu / \lambda_{max} }{s + 1 + (2\mu / \lambda_{max}) } = h(s),
\end{align}
where, for the sake of clarity, we have expressed matrix inversion in the form of division. Now, if we develop the geometric series of $h(\widetilde{\mathcal{S}})$, we obtain the following recursive expression that converges for all $\mu > 0$: 
\begin{equation}\label{proposed.filter.recursive.eq}
    h(\widetilde{\mathcal{S}}) = \left(\frac{2\mu}{2\mu + \lambda_{max}} \right) \sum_{t = 0}^{\infty} \left( { - \frac{\lambda_{max}}{2\mu + \lambda_{max}}} \right)^t \widetilde{\mathcal{S}}^t
\end{equation}
Eq. (\ref{proposed.filter.recursive.eq}) highlights our restriction to undirected graphs: if the spectrum is complex, then the recursion is not guaranteed to converge. We thus leave the extension of (\ref{proposed.filter.recursive.eq}) to directed graphs as future work. Finally, by setting
\begin{equation}
    \rho = \frac{2\mu}{2\mu + \lambda_{max}},
\end{equation}
\begin{equation}
    \psi = - \frac{\lambda_{max}}{2\mu + \lambda_{max}},
\end{equation}
and applying (\ref{proposed.filter.recursive.eq}) to $y$, we obtain the following recursive equation
\begin{equation}\label{proposed.recursive.pi.eq}
    \widetilde{p}^{(t)} = \rho y + \psi \widetilde{\mathcal{S}} \widetilde{p}^{(t-1)} 
\end{equation}
that converges to $\widetilde{\pr}_{\mu}(y)$ as $t \to \infty$. Clearly, Eq. (\ref{proposed.recursive.pi.eq}) allows us to set $\widetilde{p}^{(0)} = \pr_{\mu}(y)$ and drive the trajectory of the recursion close to the fixed point, reducing the number of iterations towards convergence.

However, Eq. (\ref{proposed.recursive.pi.eq}) is not fully satisfactory. Firstly, using it to update a PageRank vector involves sending messages across the entire network due to the fact that $\widetilde{p}^{(0)} = \pr_{\mu}(y)$ is dense, even though the update mostly takes place in the perturbed graph region. Secondly, it follows power-iteration convergence speed, which is slow. To amend these issues, we extend the result of \cite{Yoon_2018_Fast} in Eq. (\ref{pr.powerit.update}) to a large family of PageRank methods by means of the following Lemma. For the sake of notation clarity, we refer to the convergent state of (\ref{proposed.recursive.pi.eq}) by $\widetilde{\pr}_{\rho, \psi}(y) = \widetilde{p}^{(\infty)}$.

\begin{lemma}\label{proposed.lemma}
Given a fixed set of coefficients $\rho$, $\psi$, and initial condition $y$, we have that
\begin{equation}\label{proposed.lemma.update.eq}
\widetilde{\pr}_{\rho, \psi}(y) = \pr_{\rho, \psi}(y) + \frac{1}{\rho} \widetilde{\pr}_{\rho, \psi}(r) 
\end{equation} 
where 
\begin{equation}
r = \psi \left[ \widetilde{\mathcal{S}} - \mathcal{S}\right] \pr_{\rho, \psi}(y)
\end{equation}
\end{lemma}
\begin{proof}
We start with a warm restart in recursion (\ref{proposed.recursive.pi.eq}) as follows
\begin{align}
\widetilde{p}^{(1)} &= \rho y + \psi \widetilde{ \mathcal{S} } \widetilde{p}^{(0)} \\
 &= \rho y + \psi \widetilde{ \mathcal{S} } \pr_{\rho, \psi}(y) \\
 &= \pr_{\rho, \psi}(y) - \psi \mathcal{S} \pr_{\rho, \psi}(y) + \psi \widetilde{ \mathcal{S} } \pr_{\rho, \psi}(y) \\
 &=\pr_{\rho, \psi}(y) + \psi \left[ \widetilde{ \mathcal{S} } - \mathcal{S} \right] \pr_{\rho, \psi}(y) \\
 &=\pr_{\rho, \psi}(y) + r
\end{align}
Then, for the second iteration we have
\begin{align}
\widetilde{p}^{(2)} &= \rho y + \psi \widetilde{ \mathcal{S} } \widetilde{p}^{(1)} \\
 &= \rho y + \psi \widetilde{ \mathcal{S} }  \left( \pr_{\rho, \psi}(y) + r \right) \\
  &= \pr_{\rho, \psi}(y) - \psi \mathcal{S} \pr_{\rho, \psi}(y) + \psi \widetilde{ \mathcal{S} }  \left( \pr_{\rho, \psi}(y) + r \right) \\
    &= \pr_{\rho, \psi}(y) + \psi  \left[ \widetilde{ \mathcal{S} } - \mathcal{S} \right] \pr_{\rho, \psi}(y) +   \psi \widetilde{ \mathcal{S} } r \\
 &= \rho y + \psi \widetilde{ \mathcal{S} } \pr_{\rho, \psi}(y) + \psi \widetilde{ \mathcal{S} } r \\
  &= \pr_{\rho, \psi}(y) + \psi r + \psi \widetilde{ \mathcal{S} } r 
\end{align}
By successive applications of this procedure, we have that 
\begin{align}
\widetilde{p}^{(\infty)} &=  \pr_{\rho, \psi}(y) + \sum_{t = 0}^{\infty} \psi^t \widetilde{ \mathcal{S} }^t r \\
 &=  \pr_{\rho, \psi}(y) + \frac{1}{\rho}\widetilde{pr}_{\rho, \psi}(r)
\end{align}
\end{proof}
\begin{algorithm}[t!]
\caption{Distributed local updating of general PageRank methods via Chebyshev polynomials}
\label{LocalChebyUpdating.alg}
\begin{algorithmic}
        \State \textbf{Input}: $\mu$, $\pr_{\rho, \psi}(y)$, $\mathcal{R}$, $\widetilde{\mathcal{R}}$ and $\rho$, $\psi$, $\phi$.  
        \State \textbf{Output}: $\widetilde{\pr}_{\rho, \psi}(y)$. \vspace{5pt}
        \State $r =  \psi \phi [ \widetilde{\mathcal{R}} - \mathcal{R} ] \pr_{\rho, \psi}(y)$
        \State Compute $\widetilde{pr}_{\rho, \psi}(r)$ via Algorithm \ref{ChebyshevPoly.alg} 
        \State Return $\widetilde{pr}_{\rho, \psi}(y) = \pr_{\rho, \psi}(y) + \widetilde{pr}_{\rho, \psi}(r)/\rho $
\end{algorithmic}
\end{algorithm}
\noindent Lemma \ref{proposed.lemma} has several implications. Firstly, it states that updating a PageRank vector amounts to computing another PageRank vector with an initial distribution $r$ that is completely localized (non-zero) in the 1-hop vicinity of the nodes that changed between $\mathcal{G}$ and $\widetilde{\mathcal{G}}$. Secondly, since computing $\widetilde{\pr}_{\rho, \psi}(r)$ involves diffusing $r$ through the graph, then the locality of $r$ implies that only a few messages are enough to make the information necessary for an update reach the affected nodes. Thirdly, it makes it obvious that it is not necessary to use the slow recursive equation (\ref{proposed.recursive.pi.eq}) to perform the update. Instead, $\widetilde{\pr}_{\rho, \psi}(r)$ can be more efficiently computed by means of the Chebyshev polynomials. Therefore, our proposed algorithm consists in leveraging Eq. (\ref{proposed.lemma.update.eq}) and in approximating $\widetilde{\pr}_{\rho, \psi}(r)$ by means of Chebyshev polynomials. It is summarized in Algorithm \ref{LocalChebyUpdating.alg}. 

\section{Numerical Evaluation}\label{sec4}
{\bf Goals}. In this section, we evaluate the performance of the proposed algorithm \footnote{Code available at \url{https://github.com/estbautista/PageRank_Updating_Chebyshev_Paper}}. In particular, our goals are: (\emph{i}) to assess the performance gains obtained by the algorithm with respect to computing PageRank from scratch using the Chebyshev polynomials; (\emph{ii}) to demonstrate that the proposed algorithm can be used to update both standard and generalized PageRank vectors; (\emph{iii}) to assess how the performance of the algorithm degrades as perturbations grow in size; (\emph{iv}) to compare the proposed algorithm with the state-of-the-art alternatives; and (\emph{v}) to evaluate the performance of the algorithm in a tracking scenario where a PageRank vector needs to be updated during a long period of time.

{\bf Metrics}. We assess performance in terms of the number of messages that need be exchanged in order to approximate the evolved PageRank vector within a specified relative error ($\ell_2$-norm sense). We use the number of messages rather than the routine's running time because we consider this metric to better capture the complexity of a local and distributed algorithm.

{\bf Data}. We perform our experiments in the Tech-AS-Topology temporal network \cite{Rossi_2015_Network}, which is a real-world network from an autonomous system with 34.8K nodes and 171.4K edges organized in 32.8K graph snapshots. This network contains both sparse and dense regions, meaning that perturbations can affect small or large regions. For the experiments, we pre-process the data by turning the snapshots into undirected graphs, resulting in a total of 215.4K timestamped edges. The first graph snapshot in the sequence contains 32K nodes and 111.6K edges. Then, in successive snapshots, new edges adhere into the network branching nodes already present or new nodes joining the graph. The dataset only contains edge additions, therefore we simulate an edge deletion setting (see experiment 4) by reversing the time axis: we consider the originally last snapshot as the new first one and the originally first snapshot as the new last one. From this perspective, a new snapshot causes edges to disappear or nodes to leave. 

\subsection{Experiment 1}
In our first experiment, we address goals (\emph{i}) and (\emph{ii}). For this, we fix a small graph perturbation. Then, as we vary the allowed number of messages, we measure how well Algorithm \ref{LocalChebyUpdating.alg} approximates the true PageRank of the evolved graph. For comparison purposes, we apply the same test to a computation from scratch using Algorithm \ref{ChebyshevPoly.alg}. To show that the proposed algorithm can update generalized PageRank propositions as well as standard PageRank, we employ it to update standard PageRank vectors and the recently proposed $L^\gamma$-PageRank vectors from \cite{Bautista_2019_Gamma}. For this experiment, we use the first snapshot from the Tech-AS-Topology network as initial graph. Then, we use the second snapshot of the network as the perturbation: it contains 120 new edges and 1 new node joining the graph. We choose a vertex at random and use its indicator function as the initial distribution $y$ (a common setting in local graph clustering). We use $\alpha = 0.5$, measure relative error in the $\ell_2$ sense, and repeat the experiment for 20 realizations of $y$. 

Results of Experiment 1 are displayed in Figure \ref{fig_exp1}. The left panel shows the result of updating standard PageRank, while the right panel depicts the result of updating the generalized $L^\gamma$-PageRank \cite{Bautista_2019_Gamma}. They show that the proposed algorithm successfully updates both standard and generalized PageRank vectors. In both cases, the proposed algorithm offers significant approximation improvements compared to computing from scratch. We additionally verify that our updating algorithm converges at the same rate than the Chebyshev polynomials from scratch. This amends a tradeoff that needs to be made with current updating algorithms: they are a good option if only few iterations are allowed but their slow convergence makes them worse than Chebyshev polynomials from scratch if several iterations are needed \cite{Bautista_2019_Laplacian}. Lastly, we notice that the error bars (standard error) are negligible, indicating that the performance of the algorithm is irrespective of the choice of initial seed $y$ and, consequently, of a particular PageRank vector to update. 

\subsection{Experiment 2}
\begin{figure}[t!]
	\centering
	\subfloat{\includegraphics[width=0.24\textwidth]{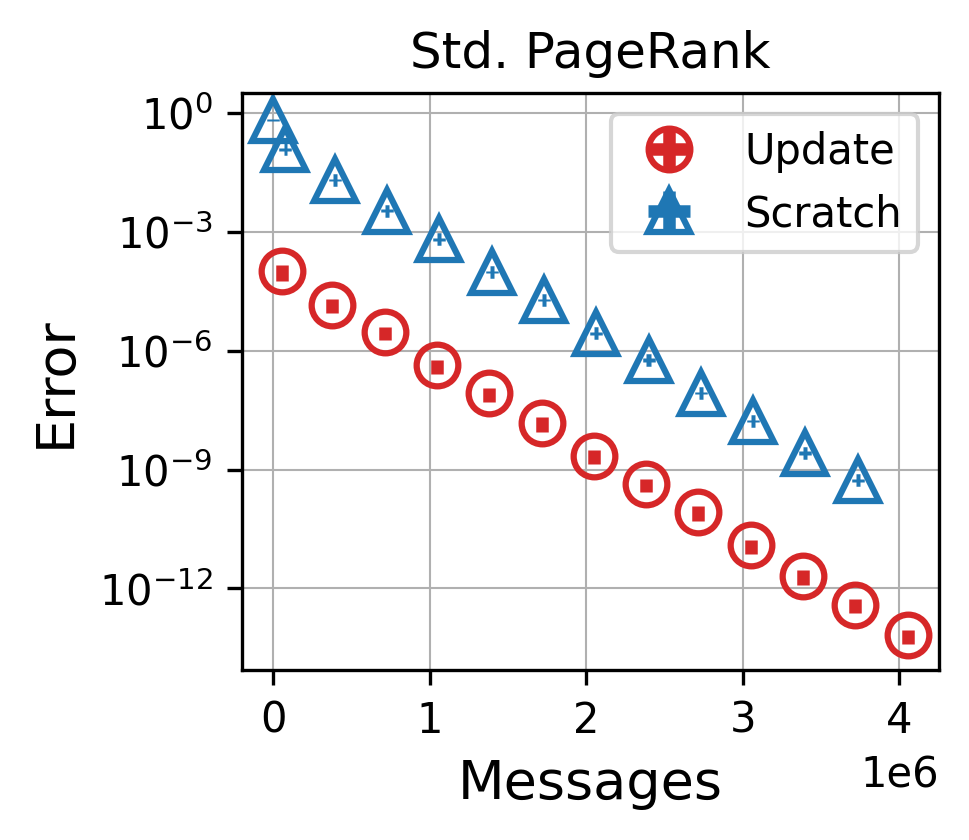}} \label{fig1a}
	\hfill
	\subfloat{\includegraphics[width=0.24\textwidth]{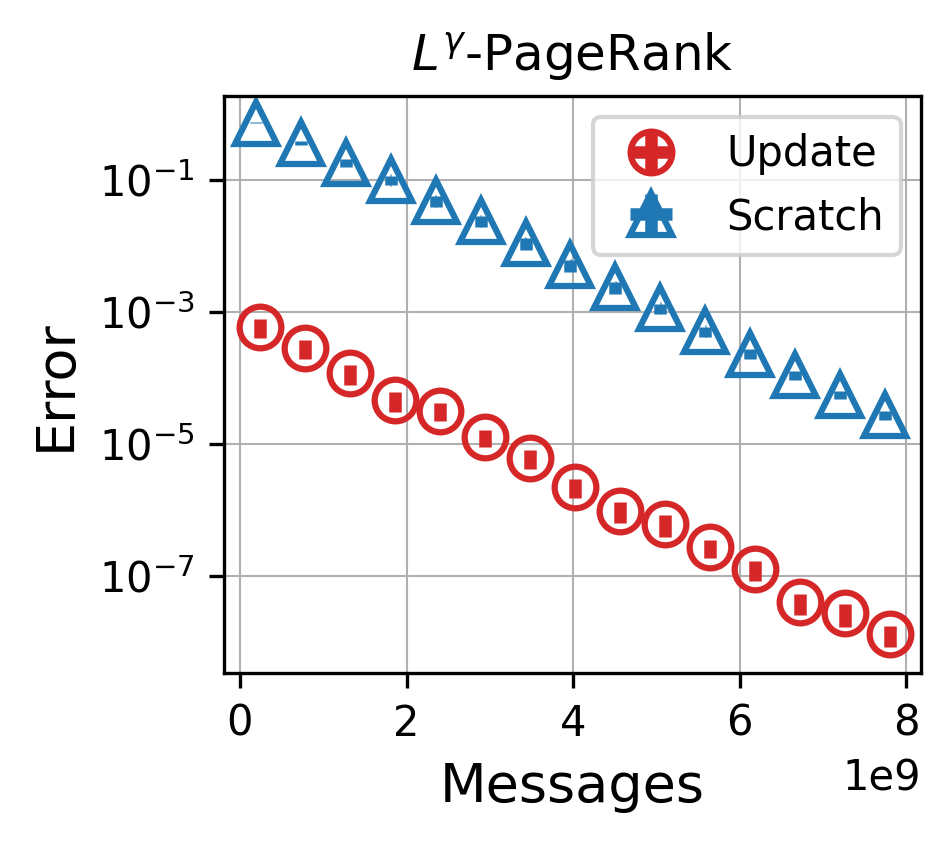}}
	\caption{Experiment 1. \small{Performance of proposed algorithm vs a computation from scratch. We assess approximation error as larger computational budgets are given to the algorithms. Left: updating standard PageRank. Right: updating a generalized PageRank.}}\label{fig_exp1}
\end{figure}
\begin{figure}[t!]
	\centering
	\includegraphics[width=0.49\textwidth]{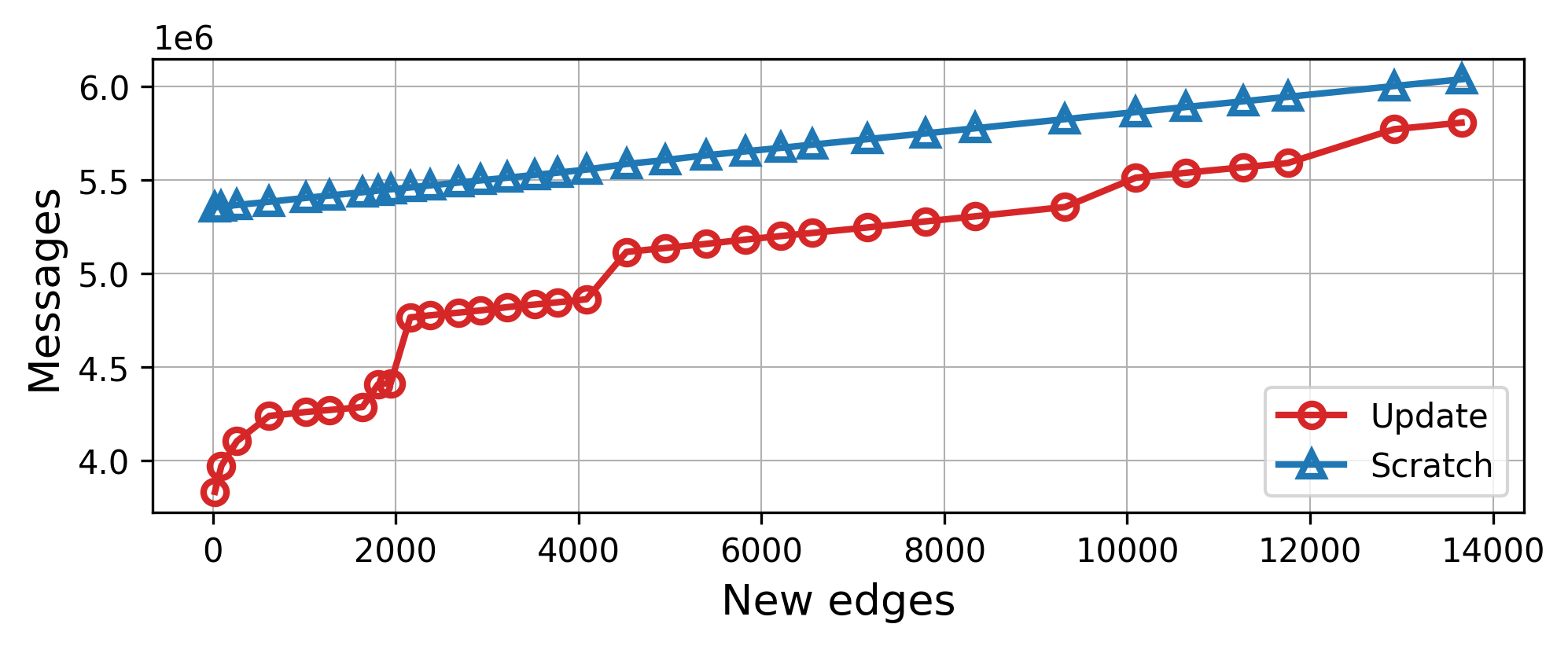}
	\caption{Experiment 2. \small{Sensitivity of proposed algorithm to perturbations. We assess the number of messages required to update within an error of $1\times10^{-13}$ as the number of added edges grows. The alternative from scratch (same conditions) is used as reference.}}\label{fig_exp2}
\end{figure}
In our second experiment, we address goal (\emph{iii}). For this, we fix a target approximation error. Then, as we vary the perturbation size, we measure the number of messages required by Algorithm \ref{LocalChebyUpdating.alg} to approximate the true PageRank of the evolved graph within the specified error. For comparison purposes, we apply the same test to a computation from scratch using Algorithm \ref{ChebyshevPoly.alg}. Since larger perturbations imply larger updates, our proposed algorithm should be highly sensitive to the size of perturbations. On the other hand, a computation from scratch should only augment messages proportionally to $\vert \widetilde{\mathcal{E}} \vert - \vert \mathcal{E}\vert$. Therefore, we aim to empirically spot the point where the update needed is so large that our updating algorithm offers no benefit over a computation from scratch. For this experiment, we use the first snapshot from the Tech-AS-Topology network as initial graph. Then, we control the size of the perturbation by aggregating an increasingly larger number of subsequent snapshots. We set $y$ as the indicator function of a random vertex and update its associated standard PageRank vector. We use $\alpha = 0.5$ and set the error at $1\times10^{-13}$.

Results of Experiment 2 are displayed in Figure \ref{fig_exp2}. For small perturbations, the number of messages needed by our algorithm to attain the desired error is small. This number increases as the perturbation grows in size, reaching a point where the updating algorithm does not provide any advantage with respect to a computation from scratch. For the Tech-AS-Topology network, this operational limit occurs for a perturbation of around 4000 new edges, which corresponds to roughly $3\%$ of the edges from the initial graph. This confirms that our updating algorithm should preferably be used when the changes in the graph are small. 

\subsection{Experiment 3}
\begin{figure}[t!]
	\centering
	\includegraphics[width=0.49\textwidth]{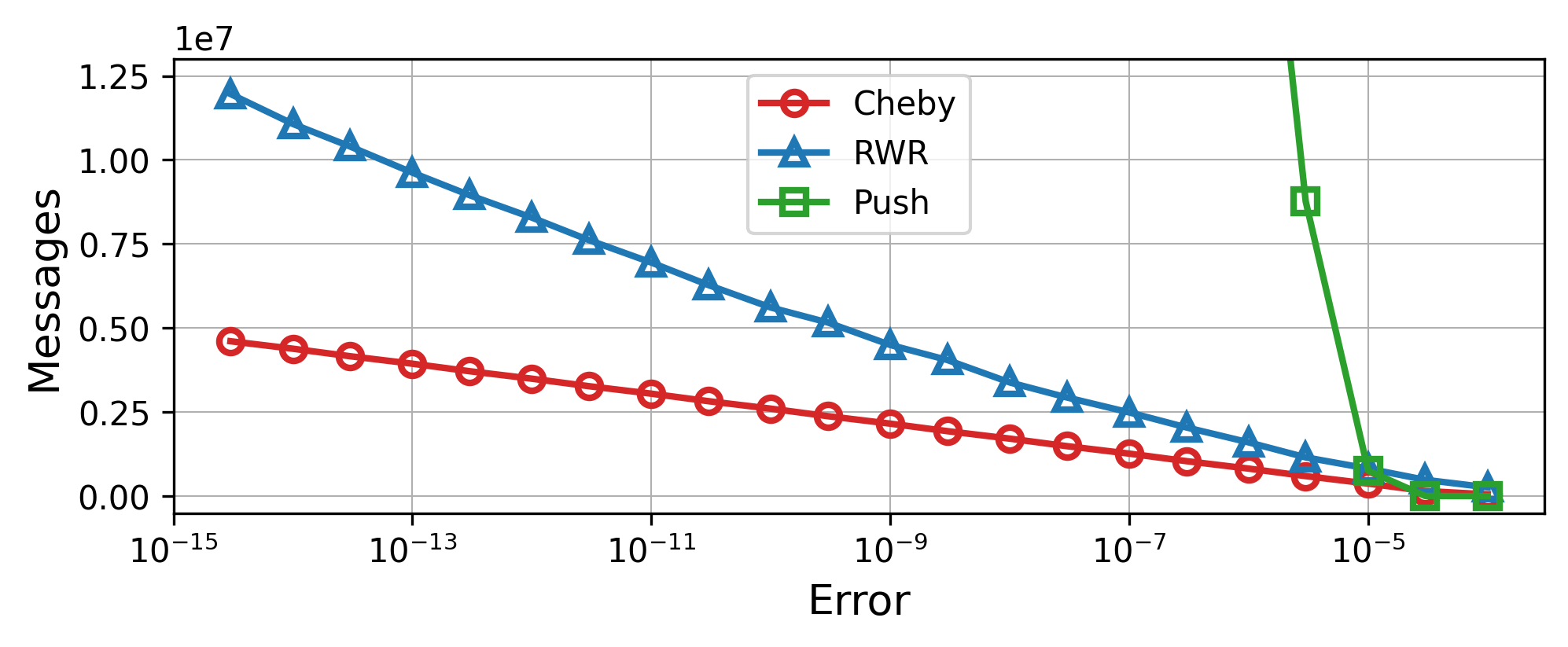}
	\caption{Experiment 3. \small{Comparison of proposed algorithm with state-of-the-art alternatives. We assess the number of messages required to update a PageRank vector within an increasingly smaller approximation error.}}\label{fig_exp3}
\end{figure}
\begin{figure}[t!]
	\centering
	\includegraphics[width=0.49\textwidth]{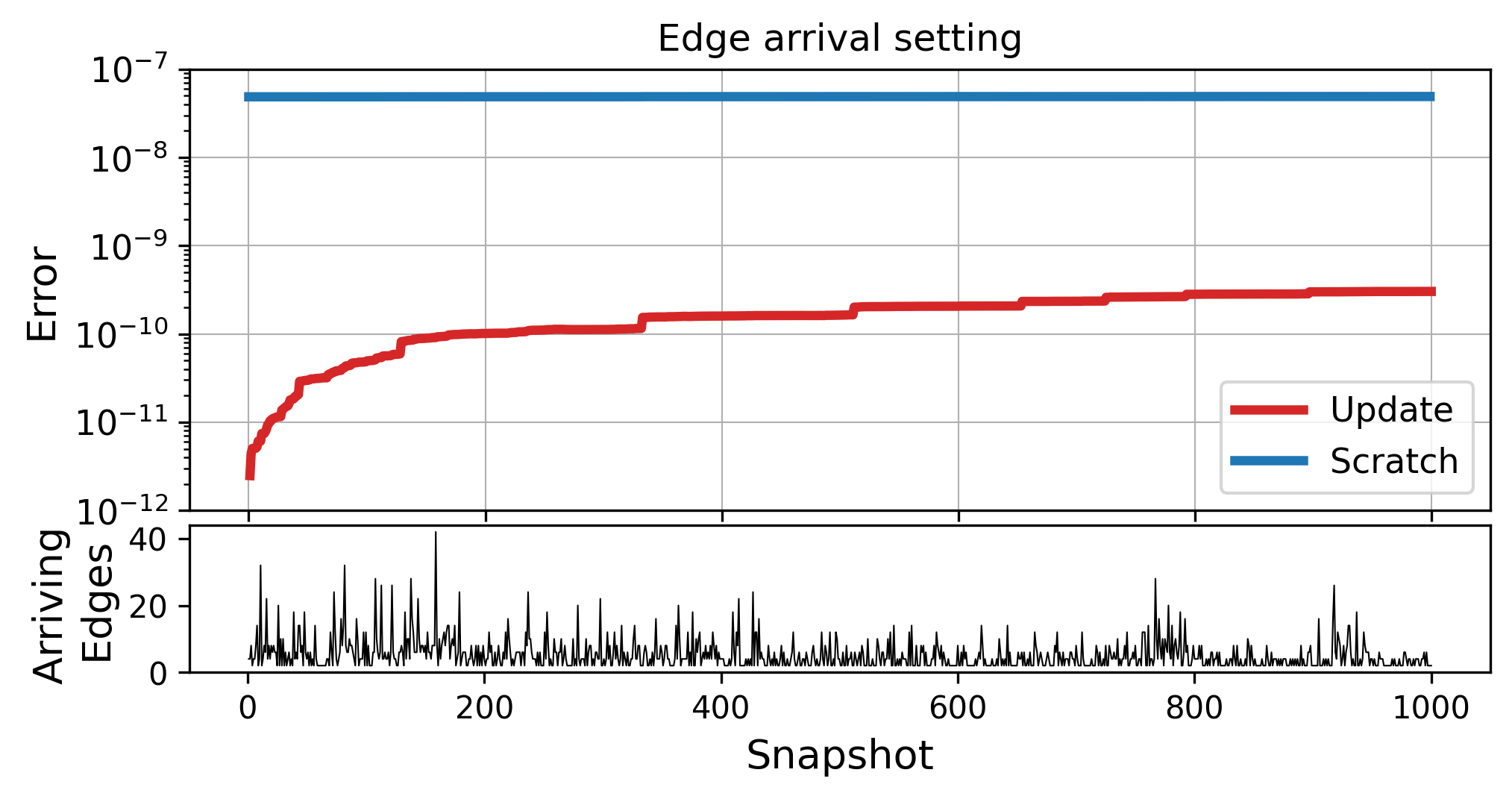}
	\includegraphics[width=0.49\textwidth]{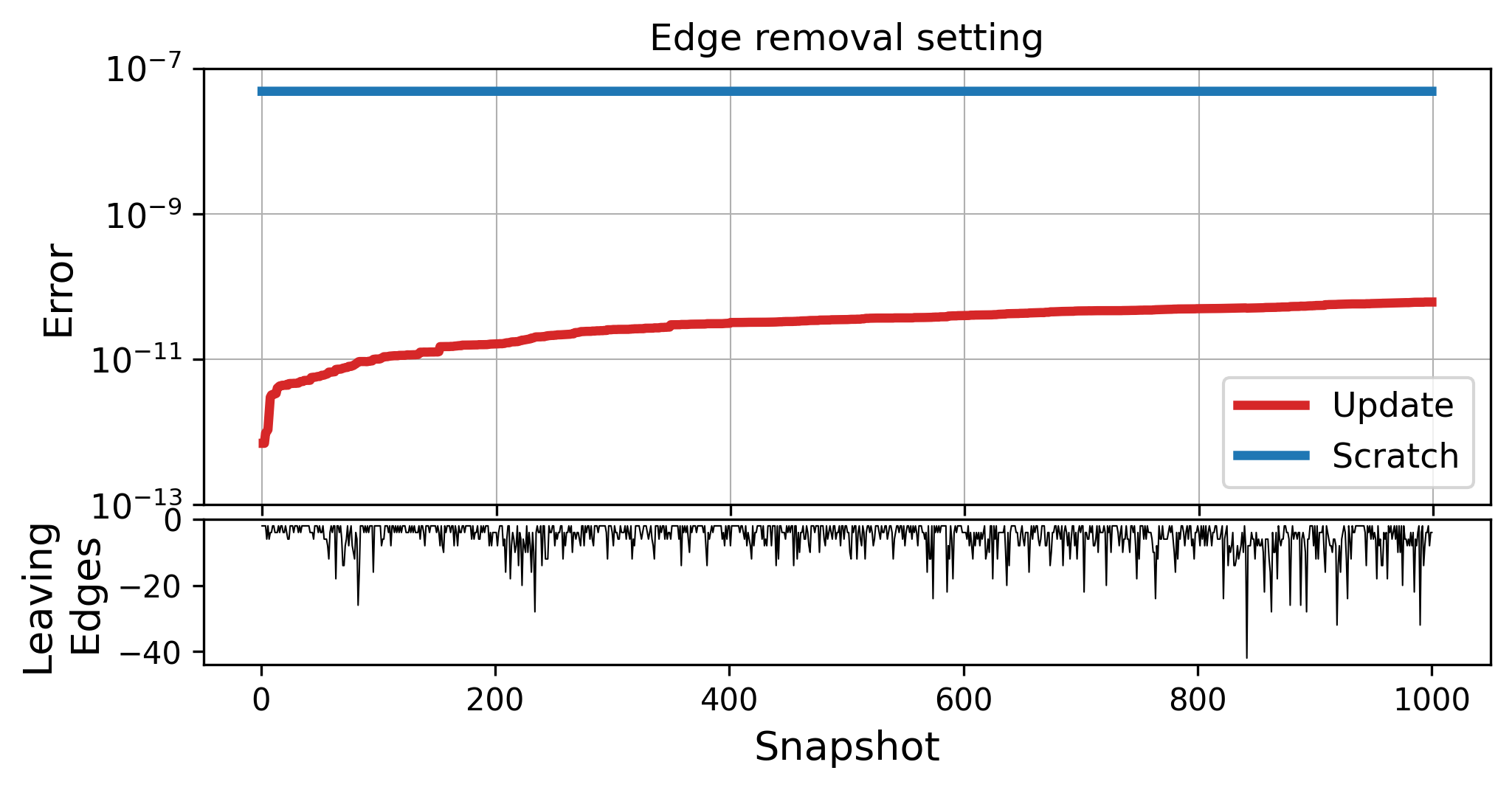}
	\caption{Experiment 4. \small{Performance of proposed algorithm in a tracking setting. We update a PageRank vector during a long period of time using $K = 15$ fixed communication rounds. We measure error with respect to the exact PageRank of each snapshot. The alternative from scratch (same conditions) is used as reference. Top: Edge arrival setting. Bottom: Edge removal setting.}}\label{fig_exp4}
\end{figure}
In our next third experiment, we address goal (\emph{iv}). For this, we fix a graph perturbation. Then, as we vary the target approximation error, we measure the number of message exchanges our proposed Algorithm \ref{LocalChebyUpdating.alg} and the state-of-the-art alternatives \cite{Yoon_2018_Fast, Ohsaka_2015_Efficient} require to approximate the true PageRank of the evolved graph within the specified error. We stress that our algorithm and the RWR one \cite{Yoon_2018_Fast} are distributed and can thus be assessed in terms of transmitted messages. Yet, the push algorithm \cite{Ohsaka_2015_Efficient} is centralized and normally studied in terms of push operations rather than messages. We notice that the complexity of each push operation is dominated by a matrix-vector multiplication that can be interpreted as transmitted messages, thus we track this quantity for the Push method. For this experiment, we use the first snapshot from the Tech-AS-Topology network as initial graph. Then, we use the second snapshot of the network as the perturbation. We set $y$ as the indicator function of a random vertex and update its associated standard PageRank vector. We use $\alpha = 0.5$.

Results of Experiment 3 are displayed in Figure \ref{fig_exp3}. The proposed algorithm outperforms state-of-the-art alternatives for local PageRank updating, being able to attain any desired relative approximation error with significantly less messages. Indeed, for approximation errors in the order of $10^{-14}$, which are the best approximations we can obtain using Python's float64 data types, the proposed Algorithm requires roughly $35\%$ less messages than the second best method of RWR. We notice that the push algorithm is not competitive for very precise approximations, as the number of message operations it requires quickly becomes large.

\subsection{Experiment 4}

In our fourth experiment, we address goal (\emph{v}). For this, we fix a number of communication rounds (K). Then, as new graph snapshots arrive, we estimate the PageRank vector of the current snapshot by updating the PageRank vector estimated for the previous snapshot using Algorithm \ref{LocalChebyUpdating.alg}. To demonstrate that our algorithm addresses equally edge additions and deletions, we run the experiment in both settings. Since the data only contains edge additions, we simulate deletions by reversing the time axis, meaning that we start from the evolved network and run backwards to the primitive one. For comparison purposes, we estimate the exact PageRank of the current snapshot via a computation from scratch using Algorithm \ref{ChebyshevPoly.alg} under the same number of communication rounds (K). We stress that this is an extremely challenging task for the updating algorithm because the vector to update is no longer the exact PageRank vector of the previous snapshot but an approximation of it, thus meaning that errors accumulate over time. Therefore, we aim to empirically spot if our method can maintain a PageRank vector for a long time or if it soon becomes worse than the approximation obtained by the method from scratch. For this experiment, we use the aggregated first 100 snapshots from the Tech-AS-Topology network as initial graph. Then, we track the standard PageRank vector during the following 1000 snapshots (reverse for edge removals). The only exact PageRank vector is given to the initial graph. We set $y$ as the indicator function of a random vertex, $K = 15$ and $\alpha = 0.5$.

Results of Experiment 4 are displayed in Figure \ref{fig_exp4}. The upper panel shows the relative error between the tracked vectors and the true PageRank of each snapshot when edge additions are considered, while the bottom panel depicts the same quantities for the case of edge deletions. The size of the perturbation in each new snapshot is shown as additional information in both panels. In both cases, the proposed algorithm is able to effectively track the PageRank vector during the entire time horizon. For early times, the updating algorithm returns extremely precise approximations: up to five orders of magnitude improvement with respect to a computation from scratch with the same computational budget. Then, we notice that errors steadily accumulate. However, it is at a sufficiently slow rate that, during the 1000 snapshots, the tracked vector is at least two orders of magnitude closer to the exact PageRank than the alternative from scratch.

\section{Conclusion}\label{sec5}
We proposed a Chebyshev polynomial-based distributed algorithm for local PageRank updating. We showed that the proposed algorithm has faster convergence than state-of-the-art alternatives, bringing us closer to the goal of effortlessly maintaining PageRank vectors in real world networks. Additionally, the algorithm can be used to update more general formulations of PageRank. These improvements were possible due to a novel updating equation that encompasses a family of PageRank formulations and that makes it direct to employ Chebyshev polynomials to locally solve the updating challenge. Numerical evaluations showed that the proposed algorithm is an effective tool for tracking PageRank vectors for a long period of time when changes in the graph are small. An interesting prospective work would be to extend these results to undirected graphs that have non-real eigenvalues and to the case in which PageRank parameters also change over time.

\section{Acknowledgements}
This work is funded in part by the ANR (French National Agency of Research) under the Limass (ANR-19-CE23-0010) and FiT LabCom grants. The authors would also like to thank P. Abry and P. Gon\c{c}alves for helpful discussions.
\bibliographystyle{ieeetr}
\bibliography{Biblio}

\begin{thebibliography}{10}

\bibitem{Page_1999_Pagerank}
L.~Page, S.~Brin, R.~Motwani, and T.~Winograd, ``The pagerank citation ranking:
  Bringing order to the web.,'' tech. rep., Stanford InfoLab, 1999.

\bibitem{Ding_2003_PageRank}
C.~Ding, X.~He, P.~Husbands, H.~Zha, and H.~Simon, ``Pagerank, hits and a
  unified framework for link analysis,'' in {\em Proceedings of the 2003 SIAM
  International Conference on Data Mining}, pp.~249--253, SIAM, 2003.

\bibitem{Haveliwala_2003_topic}
T.~H. Haveliwala, ``Topic-sensitive pagerank: A context-sensitive ranking
  algorithm for web search,'' {\em IEEE transactions on knowledge and data
  engineering}, vol.~15, no.~4, pp.~784--796, 2003.

\bibitem{Chung_2009_Local}
F.~Chung, ``A local graph partitioning algorithm using heat kernel pagerank,''
  {\em Internet Mathematics}, vol.~6, no.~3, pp.~315--330, 2009.

\bibitem{Andersen_2006_Local}
R.~Andersen, F.~Chung, and K.~Lang, ``Local graph partitioning using pagerank
  vectors,'' in {\em 2006 47th Annual IEEE Symposium on Foundations of Computer
  Science (FOCS'06)}, pp.~475--486, IEEE, 2006.

\bibitem{Tabrizi_2013_Personalized}
S.~A. Tabrizi, A.~Shakery, M.~Asadpour, M.~Abbasi, and M.~A. Tavallaie,
  ``Personalized pagerank clustering: A graph clustering algorithm based on
  random walks,'' {\em Physica A: Statistical Mechanics and its Applications},
  vol.~392, no.~22, pp.~5772--5785, 2013.

\bibitem{Avrachenkov_2012_Classification}
K.~Avrachenkov, P.~Gon{\c{c}}alves, A.~Legout, and M.~Sokol, ``Classification
  of content and users in bittorrent by semi-supervised learning methods,'' in
  {\em 2012 8th International Wireless Communications and Mobile Computing
  Conference (IWCMC)}, pp.~625--630, IEEE, 2012.

\bibitem{Merkurjev_2018_Semi}
E.~Merkurjev, A.~L. Bertozzi, and F.~Chung, ``A semi-supervised heat kernel
  pagerank mbo algorithm for data classification,'' {\em Communications in
  Mathematical Sciences}, vol.~16, no.~5, pp.~1241--1265, 2018.

\bibitem{Dostal_2014_Exploration}
M.~Dostal, M.~Nykl, and K.~Je{\v{z}}ek, ``Exploration of document
  classification with linked data and pagerank,'' in {\em Intelligent
  Distributed Computing VII}, pp.~37--43, Springer, 2014.

\bibitem{Avrachenkov_2012_Generalized}
K.~Avrachenkov, A.~Mishenin, P.~Gon{\c{c}}alves, and M.~Sokol, ``Generalized
  optimization framework for graph-based semi-supervised learning,'' in {\em
  Proceedings of the 2012 SIAM International Conference on Data Mining},
  pp.~966--974, SIAM, 2012.

\bibitem{Fontugne_2019_BGP}
R.~Fontugne, E.~Bautista, C.~Petrie, Y.~Nomura, P.~Abry, P.~Gon{\c{c}}alves,
  K.~Fukuda, and E.~Aben, ``Bgp zombies: An analysis of beacons stuck routes,''
  in {\em International Conference on Passive and Active Network Measurement},
  pp.~197--209, Springer, 2019.

\bibitem{Yoon_2019_Fast}
M.~Yoon, B.~Hooi, K.~Shin, and C.~Faloutsos, ``Fast and accurate anomaly
  detection in dynamic graphs with a two-pronged approach,'' in {\em
  Proceedings of the 25th ACM SIGKDD International Conference on Knowledge
  Discovery \& Data Mining}, pp.~647--657, 2019.

\bibitem{Yao_2012_Anomaly}
Z.~Yao, P.~Mark, and M.~Rabbat, ``Anomaly detection using proximity graph and
  pagerank algorithm,'' {\em IEEE Transactions on Information Forensics and
  Security}, vol.~7, no.~4, pp.~1288--1300, 2012.

\bibitem{AlJanabi_2019_Recommendation}
S.~Al\_Janabi and N.~Kadiam, ``Recommendation system of big data based on
  pagerank clustering algorithm,'' in {\em International Conference on Big Data
  and Networks Technologies}, pp.~149--171, Springer, 2019.

\bibitem{Zhang_2009_Collaborative}
Y.~Zhang, N.~Zhang, and J.~Tang, ``A collaborative filtering tag recommendation
  system based on graph,'' {\em ECML PKDD discovery challenge}, pp.~297--306,
  2009.

\bibitem{Nguyen_2015_Evaluation}
P.~Nguyen, P.~Tomeo, T.~Di~Noia, and E.~Di~Sciascio, ``An evaluation of simrank
  and personalized pagerank to build a recommender system for the web of
  data,'' in {\em Proceedings of the 24th International Conference on World
  Wide Web}, pp.~1477--1482, 2015.

\bibitem{Langville_2004_Deeper}
A.~N. Langville and C.~D. Meyer, ``Deeper inside pagerank,'' {\em Internet
  Mathematics}, vol.~1, no.~3, pp.~335--380, 2004.

\bibitem{Ipsen_2006_Mathematical}
I.~C. Ipsen and R.~S. Wills, ``Mathematical properties and analysis of google's
  pagerank,'' {\em Bol. Soc. Esp. Mat. Apl}, vol.~34, pp.~191--196, 2006.

\bibitem{Brezinski_2006_Pagerank}
C.~Brezinski and M.~Redivo-Zaglia, ``The pagerank vector: properties,
  computation, approximation, and acceleration,'' {\em SIAM Journal on Matrix
  Analysis and Applications}, vol.~28, no.~2, pp.~551--575, 2006.

\bibitem{Pretto_2002_Theoretical}
L.~Pretto, ``A theoretical analysis of google's pagerank,'' in {\em
  International Symposium on String Processing and Information Retrieval},
  pp.~131--144, Springer, 2002.

\bibitem{Bautista_2019_Gamma}
E.~Bautista, P.~Abry, and P.~Gon{\c{c}}alves, ``L$^\gamma$-pagerank for
  semi-supervised learning,'' {\em Applied Network Science}, vol.~4, no.~1,
  pp.~1--20, 2019.

\bibitem{DeNigris_2017_Fractional}
S.~De~Nigris, E.~Bautista, P.~Abry, K.~Avrachenkov, and P.~Gon{\c{c}}alves,
  ``Fractional graph-based semi-supervised learning,'' in {\em 2017 25th
  European Signal Processing Conference (EUSIPCO)}, pp.~356--360, IEEE.

\bibitem{Mai_2017_Counterintuitive}
X.~Mai and R.~Couillet, ``The counterintuitive mechanism of graph-based
  semi-supervised learning in the big data regime,'' in {\em 2017 IEEE
  International Conference on Acoustics, Speech and Signal Processing
  (ICASSP)}, pp.~2821--2825, IEEE, 2017.

\bibitem{Zhou_2011_Semi}
X.~Zhou and M.~Belkin, ``Semi-supervised learning by higher order
  regularization,'' in {\em Proceedings of the fourteenth international
  conference on artificial intelligence and statistics}, pp.~892--900, JMLR
  Workshop and Conference Proceedings, 2011.

\bibitem{Girault_2014_semi}
B.~Girault, P.~Gon{\c{c}}alves, E.~Fleury, and A.~S. Mor, ``Semi-supervised
  learning for graph to signal mapping: A graph signal wiener filter
  interpretation,'' in {\em 2014 IEEE International Conference on Acoustics,
  Speech and Signal Processing (ICASSP)}, pp.~1115--1119, IEEE, 2014.

\bibitem{Haveliwala_1999_Efficient}
T.~Haveliwala, ``Efficient computation of pagerank,'' tech. rep., Stanford,
  1999.

\bibitem{Kamvar_2004_Adaptive}
S.~Kamvar, T.~Haveliwala, and G.~Golub, ``Adaptive methods for the computation
  of pagerank,'' {\em Linear Algebra and its Applications}, vol.~386,
  pp.~51--65, 2004.

\bibitem{Fujiwara_2012_Efficient}
Y.~Fujiwara, M.~Nakatsuji, T.~Yamamuro, H.~Shiokawa, and M.~Onizuka,
  ``Efficient personalized pagerank with accuracy assurance,'' in {\em
  Proceedings of the 18th ACM SIGKDD international conference on Knowledge
  discovery and data mining}, pp.~15--23, 2012.

\bibitem{Bahmani_2011_Fast}
B.~Bahmani, K.~Chakrabarti, and D.~Xin, ``Fast personalized pagerank on
  mapreduce,'' in {\em Proceedings of the 2011 ACM SIGMOD International
  Conference on Management of data}, pp.~973--984, 2011.

\bibitem{Maehara_2014_Computing}
T.~Maehara, T.~Akiba, Y.~Iwata, and K.-i. Kawarabayashi, ``Computing
  personalized pagerank quickly by exploiting graph structures,'' {\em
  Proceedings of the VLDB Endowment}, vol.~7, no.~12, pp.~1023--1034, 2014.

\bibitem{Avrachenkov_2007_Monte}
K.~Avrachenkov, N.~Litvak, D.~Nemirovsky, and N.~Osipova, ``Monte carlo methods
  in pagerank computation: When one iteration is sufficient,'' {\em SIAM
  Journal on Numerical Analysis}, vol.~45, no.~2, pp.~890--904, 2007.

\bibitem{Berkhin_2006_Bookmark}
P.~Berkhin, ``Bookmark-coloring algorithm for personalized pagerank
  computing,'' {\em Internet Mathematics}, vol.~3, no.~1, pp.~41--62, 2006.

\bibitem{Ohsaka_2015_Efficient}
N.~Ohsaka, T.~Maehara, and K.-i. Kawarabayashi, ``Efficient pagerank tracking
  in evolving networks,'' in {\em Proceedings of the 21th ACM SIGKDD
  International Conference on Knowledge Discovery and Data Mining},
  pp.~875--884, 2015.

\bibitem{Bahmani_2010_Fast}
B.~Bahmani, A.~Chowdhury, and A.~Goel, ``Fast incremental and personalized
  pagerank,'' {\em Proc. VLDB Endow.}, vol.~4, p.~173–184, Dec. 2010.

\bibitem{Yoon_2018_Fast}
M.~Yoon, W.~Jin, and U.~Kang, ``Fast and accurate random walk with restart on
  dynamic graphs with guarantees,'' in {\em Proceedings of the 2018 World Wide
  Web Conference}, pp.~409--418, 2018.

\bibitem{Zhang_2016_Approximate}
H.~Zhang, P.~Lofgren, and A.~Goel, ``Approximate personalized pagerank on
  dynamic graphs,'' in {\em Proceedings of the 22nd ACM SIGKDD International
  Conference on knowledge discovery and data mining}, pp.~1315--1324, 2016.

\bibitem{yoon2020autonomous}
M.~Yoon, T.~Gervet, B.~Hooi, and C.~Faloutsos, ``Autonomous graph mining
  algorithm search with best speed/accuracy trade-off,'' {\em 2020 IEEE
  International Conference on Data Mining (ICDM)}, pp.~751--760, 2020.

\bibitem{Bautista_2019_Laplacian}
E.~Bautista~Ruiz, {\em {Laplacian Powers for Graph-Based Semi-Supervised
  Learning}}.
\newblock Theses, {Universit{\'e} de Lyon}, Nov. 2019.

\bibitem{Shuman_2011_Chebyshev}
D.~I. Shuman, P.~Vandergheynst, and P.~Frossard, ``Chebyshev polynomial
  approximation for distributed signal processing,'' in {\em 2011 International
  Conference on Distributed Computing in Sensor Systems and Workshops (DCOSS)},
  pp.~1--8, IEEE, 2011.

\bibitem{cheng_2019_iterative}
C.~Cheng, J.~Jiang, N.~Emirov, and Q.~Sun, ``Iterative chebyshev polynomial
  algorithm for signal denoising on graphs,'' in {\em 2019 13th International
  conference on Sampling Theory and Applications (SampTA)}, pp.~1--5, IEEE,
  2019.

\bibitem{tseng_2021_minimax}
C.-C. Tseng and S.-L. Lee, ``Minimax design of graph filter using chebyshev
  polynomial approximation,'' {\em IEEE Transactions on Circuits and Systems
  II: Express Briefs}, vol.~68, no.~5, pp.~1630--1634, 2021.

\bibitem{tian_2014_chebyshev}
D.~Tian, H.~Mansour, A.~Knyazev, and A.~Vetro, ``Chebyshev and conjugate
  gradient filters for graph image denoising,'' in {\em 2014 IEEE International
  Conference on Multimedia and Expo Workshops (ICMEW)}, pp.~1--6, IEEE, 2014.

\bibitem{defferrard_2016_convolutional}
M.~Defferrard, X.~Bresson, and P.~Vandergheynst, ``Convolutional neural
  networks on graphs with fast localized spectral filtering,'' {\em Advances in
  neural information processing systems}, vol.~29, pp.~3844--3852, 2016.

\bibitem{yan_2021_spatial}
B.~Yan, G.~Wang, J.~Yu, X.~Jin, and H.~Zhang, ``Spatial-temporal chebyshev
  graph neural network for traffic flow prediction in iot-based its,'' {\em
  IEEE Internet of Things Journal}, 2021.

\bibitem{Rossi_2015_Network}
R.~A. Rossi and N.~K. Ahmed, ``The network data repository with interactive
  graph analytics and visualization,'' in {\em AAAI}, 2015.

\end{thebibliography}

\end{document}